\documentclass[a4paper,UKenglish]{lipics-v2019}

\title{Dynamic Epistemic Logic Games with Epistemic Temporal Goals}

\author{Bastien~Maubert}{Universit\'a degli Studi di Napoli ``Federico~II'', Italy}{}{}{}
\author{Aniello~Murano}{Universit\'a degli Studi di Napoli ``Federico~II'', Italy}{}{}{}
\author{Sophie Pinchinat}{Universit\'e de Rennes, France}{}{}{}
\author{Francois Schwarzentruber}{ENS Rennes, France}{}{}{}
\author{Silvia Stranieri}{Universit\'a degli Studi di Napoli ``Federico~II'', Italy}{}{}{}

\authorrunning{B. Maubert, A. Murano, S. Pinchinat, F. Schwarzentruber \& S. Stranieri}

\Copyright{}


\category{}

\relatedversion{}

\supplement{}

\hideLIPIcs  

\EventEditors{}
\EventNoEds{2}
\EventLongTitle{}
\EventShortTitle{}
\EventAcronym{CVIT}
\EventYear{2016}
\EventDate{December 24--27, 2016}
\EventLocation{Little Whinging, United Kingdom}
\EventLogo{}
\SeriesVolume{42}
\ArticleNo{23}

\usepackage[utf8]{inputenc}
\usepackage{times}
\usepackage{graphicx}
\usepackage{latexsym}
\usepackage{amsmath,amssymb,stmaryrd}
\usepackage{xspace}
\usepackage{enumitem}

\usepackage{tikz}
\usetikzlibrary{arrows,patterns,topaths,automata}

\usepackage{xcolor}
\newif\ifdraft\draftfalse
\usepackage{macros}
\usepackage{draftmode}

\usepackage{environments}

\nolinenumbers

\begin{document}

\maketitle
\begin{abstract}
  Dynamic Epistemic Logic (DEL) is a logical framework in which one
  can describe in great detail how actions are perceived by the
  agents, and how they affect the world.
  DEL games were recently introduced as a way to define classes of
  games with imperfect information where the actions available to the
  players are described very precisely. This framework makes it possible to define easily, for instance, classes
  of games where players can only use public actions or public
  announcements. These games have been studied for reachability
  objectives, where the aim is to reach a situation satisfying some
  epistemic property expressed in epistemic logic; several
  (un)decidability results have been established.

  In this work we show
 that the decidability results obtained for reachability objectives
 extend to a much more general class of winning conditions, namely
 those expressible in the epistemic temporal logic LTLK. To do so we
 establish that the infinite  game structures generated by DEL public actions are regular,
 and we describe how to obtain finite representations on which we rely
 to solve them.
\end{abstract}

\section{Introduction}

\begin{table*}[]
	\begin{center}
		\begin{tabular}{|l|l|l|l|}
			\hline
			{\bf Type of DEL game presentation}& {\bf Logic}  & {\bf Complexity}   \\
			\hline
			propositional actions and hierarchical information & \LTLK  & decidable (Th.~\ref{theo-propositional}) \\
			\hline
			public actions & \LTLK  & \TWOEXPTIME-complete (Th.~\ref{theo-publicactions})\\
			\hline
			public announcements & \fragmentLTLKUnonesting   & \PSPACE-complete (Th.~\ref{theorem:objpubannoun})  \\
			\hline
		\end{tabular}
	\end{center}
	\caption{Summary of our contribution. \label{table:results}}
\end{table*}

Strategic reasoning in  multi-agent systems 
refers to a number of important issues for settings where a team of
agents  have to take
decisions in order to achieve some goals, while evolving in  an environment
that may pursue different objectives. 
Application domains are numerous (economics, robotics, distributed
computing systems, web services, etc). For instance, drones patrolling
an area may have to decide which trajectory to take so that the status
(safe or unsafe) of each zone in this area is always known to at
least one of them, while antagonistic agents try to keep the status of
some areas secret.  It is a real challenge to automatically compute adequate individual strategies for the
agents. In this work we consider the   \emph{distributed strategy
  synthesis problem}, in which a team of agents collaborates towards a common
goal, while the environment is purely antagonistic.

Because agents typically have a local view of the system, such
situations are usually modelled as imperfect information game
arenas, \ie, graphs whose nodes represent positions of the game, edges are the
possible actions, and equivalence relations capture indistinguishability of positions.  To reflect imperfect information, strategies must
prescribe the same action in indistinguishable situations; such
strategies are classically called \emph{uniform} or
\emph{observation-based}
strategies~\cite{DBLP:journals/jcss/Reif84,herzig2006knowing,DBLP:journals/iandc/BerwangerCWDH10,DBLP:journals/logcom/PileckiBJ17}. Also
the goal, or \emph{winning condition}, is often expressed in some
logical language such as \LTL~\cite{pnueli1989synthesisshort,PR90} or
\LTLK, its
extension with knowledge
operators~\cite{van1998synthesis,DBLP:conf/concur/MeydenW05}.
For the patrolling example, one could consider the \LTLK formula
$\G\bigwedge_{\text{zone}}\bigvee_{\text{agent}}(\K[\text{agent}]
\,\text{safe}_{\text{zone}} \vee \K[\text{agent}] \neg
\text{safe}_{\text{zone}})$ which says that always, for every zone,
some agent knows whether it is safe or not.

Distributed strategy
synthesis  is known to be undecidable
\cite{DBLP:conf/focs/PetersonR79,PR90}, but the numerous literature on
the topic has
identified two main decidable cases:  the case where actions in
the games are \emph{public} (known to all agents)~\cite{van1998synthesis,DBLP:conf/concur/MeydenW05,ramanujam2010communication,BelardinelliLMR17a,DBLP:conf/ijcai/BelardinelliLMR17,DBLP:conf/fossacs/Bouyer18},
and the case of \emph{hierarchical information} (the set of
agents can be totally ordered so that what is known propagates along
this order)~\cite{kupermann2001synthesizing,PR90,DBLP:journals/fmsd/GastinSZ09,
  peterson2002decision,DBLP:conf/lics/FinkbeinerS05,pinchinat2005decidable,DBLP:conf/atva/ScheweF07,berwanger2018hierarchical}.

However, the  state explosion problem makes game
structures often very large, making distributed 
synthesis intractable. In order to circumvent this difficulty, we
promote a \emph{planning approach} using Dynamic Epistemic Logic
(DEL)~\cite{DitmarschvdHoekKooi}: instead of representing explicitly
the game structure, we consider implicit descriptions by means of
\emph{DEL presentations}. These consist in a finite initial
\emph{epistemic model} that reflects the initial knowledge of the
agents, and a finite set of \emph{epistemic actions} available to them
and the other agents in the environment. Such implicit descriptions
make it easier for the modeller to add, modify or remove
actions. Also, since DEL action models were introduced to represent in
detail how events are observed by agents, this setting is very convenient
 to define various types of actions such as public actions or
semi-private announcements, and study how restricting to such actions
can make distributed synthesis easier.

While DEL presentations have been widely used in epistemic planning
(finding a finite succession of events that achieves
some epistemic property)~\cite{DBLP:journals/corr/Bolander17},
only recently have adversarial aspects been considered in this
setting, along with strategic problems such as 
 distributed strategy
 synthesis. In~\cite{DBLP:conf/ijcai/MaubertPS19},
 agents are split into two antagonistic teams
$\agentsexists$ and $\agentsforall$, and agents in
$\agentsexists$ pursue some goal while agents in $\agentsforall$ try to prevent them
 from winning (these are  zero-sum games).  For
 reachability objectives  where the team
$\agentsexists$ wants to reach a situation that satisfies some epistemic
property, 
it is shown in~\cite{DBLP:conf/ijcai/MaubertPS19} that, as in the setting of
explicit game arenas,  distributed strategy synthesis is  undecidable,
but decidability can be recovered for the case of public actions and hierarchical information. 


In this work we lift these decidability results from reachability goals to the much larger family of
winning conditions expressible in the temporal epistemic logic \LTLK,
that blends temporal operators and epistemic modalities. 
This logic can express the reachability objectives
from~\cite{DBLP:conf/ijcai/MaubertPS19} (they correspond to formulas
of the form $\F\phi$, where $\F$ is the ``Finally'' temporal operator and $\phi$ is an epistemic formula), but also
 safety, liveness properties, and many more (see for
 instance~\cite{DBLP:conf/concur/MeydenW05} for a detailed security example).

 In all our decidability results, a crucial step is to show that the
 game arena induced by a DEL game presentation, which is in general
 infinite, can be represented finitely.  This was already known for
 the case of actions whose preconditions do not involve knowledge but
 are purely propositional formulas (so-called \emph{propositional
   actions})~\cite{DBLP:phd/hal/Maubert14,AiML2018Gaetanetal}, and we
 use it to transfer an existing result for distributed synthesis in
 explicit game arenas with hierarchical information. The main
 technical contribution of this work is to prove that the infinite
 game generated from a DEL game presentation is regular also in the case of
 public actions. This is done by observing that, modulo isomorphism,
 such actions can only generate finitely many different epistemic
 models from the initial one, thus allowing us to get an equivalent finite game as the quotient of
 the infinite one. This, combined with a recent result on game arenas
 with public actions~\cite{DBLP:conf/atal/BelardinelliLMR17}, yields a
 procedure that runs in doubly exponential time, just as the case of
 \LTL games~\cite{pnueli1989synthesisshort}. Additionally, for public
 announcements (a special case of public actions) and the syntactic fragment of \LTLK without next
 operator and local knowledge properties only, we show an even stronger characteristic of game arenas that allows us to
 reduce to a polynomial-length horizon game and to derive an optimal \PSPACE procedure.


 \vskip1em
 \noindent {\bf Related work.} The only work on DEL games that we are
 aware of is~\cite{DBLP:conf/ijcai/MaubertPS19}, and it only considers
 reachability objectives. However our results relate to the many  aforementioned
 results for distributed synthesis in explicit game
 structures with either public actions or hierarchical
 information. Those dealing with epistemic temporal logic are the
 closest to ours and can be found
 in~\cite{DBLP:conf/mfcs/Puchala10,DBLP:conf/kr/MaubertM18} for
 hierarchical information,
 and~\cite{van1998synthesis,DBLP:conf/concur/MeydenW05,DBLP:conf/atal/BelardinelliLMR17,DBLP:conf/ijcai/BelardinelliLMR17}
 for public actions.


  \vskip1em
\noindent {\bf Plan.}  Section~\ref{sec-prelim} recalls games with
imperfect information and the logic \LTLK. Section~\ref{sec-del-games}
recalls DEL game presentations as defined in
~\cite{DBLP:conf/ijcai/MaubertPS19}. The central
sections~\ref{sec-propositional}, \ref{sec:publicactions},
\ref{sec:publicannouncements} describe our contributions for the cases
of propositional actions, public actions and public announcements,
respectively. We discuss our results in Section~\ref{sec:conclusion}.


\section{Preliminaries}
\label{sec-prelim}

In this section we recall basics about games with imperfect
information and the epistemic temporal logic \LTLK.

\subsection{Games with imperfect information}
\label{sec-games}

We consider multiplayer game arenas with imperfect information in the
spirit of, \eg,~\cite{schobbens2004alternating,DBLP:journals/jancl/JamrogaA07,DBLP:journals/acta/BerwangerMB18}.
Since the DEL games we define in the next section are turn-based, \ie, the
agents play in turns and not concurrently, we define
turn-based arenas instead of the more general concurrent ones
usually considered in the aforementioned works.

For the rest of the paper let us fix  a countable set of \emph{atomic
  propositions} $\AP$ and a finite set of \emph{agents}
$\agtset$ that is partitioned  into two antagonistic teams, $\agentsexists$ and $\agentsforall$.

\begin{definition}
\label{def-game-arena}
A \emph{game arena} 
$\gamestruct=(\setpos,\setposinit,\setact,\trans,\turnfunc,(\epistemicrelationequiv{\agent})_{\agent\in
  \Ag},\valuationfunction)$ is a tuple where:
\begin{itemize}
\item $\setpos$ is a non-empty set
  of \emph{positions},
\item $\setposinit\subseteq\setpos$ is the set of \emph{initial
   positions},
\item $\setact$ is a non-empty set of \emph{actions},
\item $\trans : \setpos \times \setact \rightharpoonup \setpos$ is a
partial  \emph{transition function}, 
\item $\turnfunc:\setpos\to\agtset$ is a \emph{turn function},
\item $\epistemicrelationequiv{\agent}\,\subseteq \setpos\times\setpos$
  is an \emph{indistinguishability relation} for agent $\agent$, and
\item $\valuationfunction:\setpos\to 2^\AP$ is a \emph{labelling} or
  \emph{valuation function}.
\end{itemize}
\end{definition}

In a position $\pos$, agent $\turnfunc(\pos)$ chooses an action $\act$
such that $\pos'=\trans(\pos,\act)$ is defined, and the game proceeds
similarly from position  $\pos'$. We let $\setact(\pos)=\{\act\mid
(\pos,\act)\in\dom(\trans)\}$, where $\dom$ denotes the domain, and we assume that
$\setact(\pos)\neq\emptyset$ for every position
$\pos$.

A \textit{play} $\play=\explay$ in $\gamestruct$ is
an infinite sequence of positions  such that for all $i\in \mathbb{N}$,
there exists $\act\in\setact$ such that $\pos_{i+1}=\trans(\pos_i,\act)$. We let  $\playposi=\pos_i$ and
$\play_{\leq i}=\exhistory[i]$. We also let $\setplays$ denote the set
of plays in $\gamestruct$.  A \textit{history}
$\hist=\exhistory$ is a finite nonempty prefix of a play,
$\last(\hist)=\pos_n$ is the last position in $\hist$ and
$\sethistories$ is the set of histories in $\gamestruct$. We may
write $\turnfunc(\hist)$, $\setact(\hist)$ and
$\valuationfunction(\hist)$ for, respectively, $\turnfunc(\last(\hist))$, $\setact(\last(\hist))$ and
$\valuationfunction(\last(\hist))$.

The indistinguishability relation $\epistemicrelationequiv{\agent}$ is
an equivalence relation between positions  of the
game arena  that represents how agent $\agent$ observes them:
$\pos\epistemicrelationequiv{\agent}\pos'$ means that agent $a$ 
cannot distinguish between positions $\pos$ and $\pos'$. As a result, we assume
that if $\pos\epistemicrelationequiv{\agent}\pos'$ for some agent $\agent$, then
$\turnfunc(\pos)=\turnfunc(\pos')$, which means that  agents know
whose turn it is to play. In addition, if $\turnfunc(\pos)=\agent$ and
$\pos\epistemicrelationequiv{\agent}\pos'$, we assume that
$\setact(\pos)=\setact(\pos')$, meaning that the agent who has to make a move knows
which actions are available. 

We consider agents that have \emph{synchronous perfect
  recall}, \ie they remember the whole sequence of observations they
made, and know how many moves have been made. To model this, each
indistinguishability relation $\epistemicrelationequiv{\agent}$ is lifted to histories as
follows: $\hist\epistemicrelationequiv{\agent}\hist'$ if $|\hist|=|\hist'|$
and  for every $i< |\hist|$ it holds that
 $\hist_i\epistemicrelationequiv{\agent}\hist'_i$.

A \textit{strategy} for agent $\agent$ is a partial function
$\strategy : \sethistories \rightharpoonup \setact$ such that for
every $\history$ with $\turnfunc(\history)=\agent$, it holds
that $\strategy(\history)\in\setact(\hist)$.
Because agents can only base their decisions on what they observe,
their strategies must assign the same action to indistinguishable
situations: 
a strategy $\strategy$ for agent $\agent$ is \emph{uniform}
if, for all histories $\hist,\hist'$ such that
$\turnfunc(\hist)=\turnfunc(\hist')=\agent$ and
$\hist\epistemicrelationequiv{\agent}\hist'$, it holds that $\strategy(\hist)=\strategy(\hist')$.
We say that a play $\play$ \emph{follows} a strategy $\strategy$ for
agent $\agent$ if for every $i\in\N$ such that $\turnfunc(\play_{\leq
  i})=\agent$, it holds that $\playposi[i+1]=\trans(\play_i,\strategy(\play_{\leq i}))$.
A \emph{distributed strategy} for  a group of agents $\coal\subseteq\agtset$
 is a tuple $\disstrategy[\coal]=(\strategy_\agent)_{\agent\in\coal}$, and we
 write $\outcome(\disstrategy[\coal])$ the set of \emph{outcomes} of
 $\disstrategy[\coal]$, \ie the set of plays that start in $\setposinit$ and
 follow each $\strategy_\agent$ for $\agent\in\coal$.

A \emph{game} $\game=(\gamestruct,\win)$ is a 
 game arena $\gamestruct$ with a  \emph{winning
condition} $\win\subseteq \setplays$.
Team $\agentsexists$
\emph{wins} a game $\game$  if there is a distributed strategy
$\disstrategy$ such that every play in
$\outcome(\disstrategy)$  is in $\win$.

\begin{definition}
  \label{sec-unfolding}
Let
  $\gamestruct=(\setpos,\setposinit,\setact,\trans,\turnfunc,(\epistemicrelationequiv{\agent})_{\agent\in
    \Ag},\valuationfunction)$ be a game arena. We define the \emph{unfolding} of
  $\gamestruct$ as the game arena
  $\gameunfold=(\setpos',\setposinit',\setact,\trans',\turnfunc',(\epistemicrelationequiv{\agent}')_{\agent\in
    \Ag},\valuationfunction')$ where $\setpos'=\sethistories$, $\setposinit'=\setposinit$, for
  every $\hist\in\sethistories$,
  $\turnfunc'(\hist)=\turnfunc(\last(\hist))$ and
   $\valuationfunction'(\hist)=\valuationfunction(\last(\hist))$, 
  $\epistemicrelationequiv{\agent}'$  is  the synchronous perfect-recall
lifting of $\epistemicrelationequiv{\agent}$ to histories, and
  \[\trans'(\hist,\act)=
    \begin{cases}
      \hist\cdot \pos & \mbox{if }\trans(\last(\hist),\act)=\pos \\
      \text{undefined}  & \mbox{if }\trans(\last(\hist),\act) \mbox{
        is undefined}
    \end{cases}
\]
\end{definition}

The natural bijection between plays of
$\gamestruct$ and plays of $\gameunfold$ induces a winning condition $\winunfold$ over arena
$\gameunfold$. Additionally, because of the natural bijection between
strategies in $\gamestruct$ and strategies in $\gameunfold$, $\agentsexists$ wins
$(\gamestruct,\win)$ if, and only if, $\agentsexists$ wins
$(\gameunfold,\winunfold)$. We say that two game structures
$\game$ and $\game'$ are \emph{equivalent} whenever their unfoldings are
isomorphic\footnote{Some looser notion of bisimulation between games
could also be considered but isomorphism fits here.}.

In this work we are interested in winning conditions expressed in the
logic of knowledge and time called \LTLK (standing for linear temporal logic with knowledge), which extends \LTL with
knowledge operators for each agent.

\subsection{Linear-time temporal logic with knowledge}
\label{sec-ltlk}


The syntax of \LTLK is given by the following grammar:
\begin{equation*}
\phi ::= p \mid \lnot \phi \mid \phi \vee \phi \mid
\next\phi \mid \phi \until \phi \mid \K \phi 
\end{equation*}
where $p\in\AP$ and $\agent\in\agtset$.
The formula
$\next\phi$ reads as ``at the next step, $\phi$ holds'',
$\phi\until\phi'$ reads as ``$\phi$ holds until $\phi'$ holds'', and
$\lknow\agent\phi$  is read ``agent $\agent$ knows that
$\phi$ is true''.


The
\emph{size} $|\phi|$ of a formula $\phi$ is the number of symbols in
it.

We exhibit two important syntactic fragments of \LTLK: \emph{Epistemic
  Logic} $\languageEL$  obtained by removing temporal operators
$\next$ and $\until$, \ie, generated by grammar
\[\phi ::= \atom \mid \neg\phi \mid \phi\lor\phi
\mid \lknow\agent\phi,\] 
and \emph{Propositional logic}
$\languagePropositional$ obtained by also removing the knowledge
modality.




The logic \LTLK is interpreted over a moment $i\in\N$ along a play $\play\in\setplays$ in a game
arena 
$\gamestruct=(\setpos,\setact,\trans,\turnfunc,(\epistemicrelationequiv{\agent})_{\agent\in
  \Ag},\valuationfunction)$. We write $\gamestruct,\play,i \models \phi$, and
read `formula $\phi$ holds at moment $i$ along play $\play$ of game
arena $\gamestruct$', defined inductively over $\phi$ as follows.
\begin{align*}
  \gamestruct,\play,i &\models p &&\ifdef p\in \valworlds(\play_{i})\\
  \gamestruct,\play,i &\models \lnot \phi &&\ifdef\gamestruct,\play,i \not \models \phi\\
  \gamestruct,\play,i &\models \phi_1 \vee \phi_2
                                 &&\ifdef\gamestruct,\play,i \models
                                    \phi_1 \; \mbox{ or }\;\gamestruct,\play,i \models \phi_2\\
  \gamestruct,\play,i &\models \next\phi &&\ifdef\gamestruct,\play,i+1 \models \phi\\
  \gamestruct,\play,i &\models \phi_1 \until \phi_2
                                 &&\ifdef\exists i'\geq i
                                    \;\text{ s.t.
                                    }\;\gamestruct,\play,i'
                                           \models \phi_2 \;\text{
                                    and }\\
                                          & && \quad\;\;\forall i'' \;\text{ s.t. }\; i\leq i''<i', \gamestruct,\play,i'' \models \phi_1\\
  \gamestruct,\play,i &\models \K\phi &&\ifdef \forall
                                            \play'\in\setplays \;\text{ s.t. }\;
                                            \play'_{\leq i} \relequiva
                                         \play_{\leq i}, \\
  & && \quad\;\;\gamestruct,\play',i \models \phi
\end{align*}

An \LTLK formula $\phi$ naturally denotes a winning
condition:
\[\win_\phi=\{\play\in\setplays\mid \gamestruct,\play,0\models\phi\}.\]


\section{DEL games}
\label{sec-del-games}

In this section we recall the definition of DEL games as recently
introduced in~\cite{DBLP:conf/ijcai/MaubertPS19}.
We start with definitions for  epistemic
models and DEL event models.

\subsection{The classic DEL setting}
\label{sec-EL}


In the usual possible-worlds semantics of epistemic logic, models are Kripke
structures with interpretations for atomic
propositions~\cite{Fagin95knowledge}.

\label{sec-defDEL}
\begin{definition}
	\label{definition:epistemicmodel}
	An \emph{epistemic model} $\kripkemodel = (\setworlds,
        (\epistemicrelation{\agent})_{\agent \in \agtset},
        \valuationfunction)$ is a tuple  where 
	\begin{itemize}
		\item
 $\setworlds$ is a non-empty finite set of possible \emph{worlds} (or situations),
		\item
 $\relworldsa\subseteq \setworlds\times \setworlds$
		is an
		\emph{indistinguishability relation} for agent $\agent$, and
		\item
 $\valuationfunction:\setworlds\rightarrow 2^\AP$ is a \emph{valuation function}.
	\end{itemize}
\end{definition}

\tikzstyle{agenta}=[line width=2]
\tikzstyle{agentb}=[dashed]

We may write $\world\in\model$ for $\world\in\setworlds$. 
As for games with imperfect information introduced in the previous
section, we assume that indistinguishability
relations~$\epistemicrelation{\agent}$ are equivalence
relations.
The valuation function $\valuationfunction$
defines which atomic propositions  hold in a world.  A
pair $(\kripkemodel,\world)$ where $\world\in\model$ is called a \emph{pointed epistemic
  model}, and we define $|\kripkemodel|=|\setworlds| + \sum_{\agent \in \agtset} |\epistemicrelation{\agent}|
+ \sum_{\world \in \setworlds} |\valuationfunction(\world)|$, the
\emph{size} of $\kripkemodel$. We will
only consider finite models, \ie we assume that $\setworlds$ is finite
and $\valuationfunction(\world)$ is finite for all worlds $\world$.


\begin{definition}
  We define $\kripkemodel, \world\models \formula$, read as
`formula $\formula$ holds in the pointed epistemic model
$(\kripkemodel, \world)$', by induction on $\phi$, as follows:
	\begin{itemize}
		\item $\kripkemodel,\world\models \atom$ if $\atom\in\valuationfunction(\world)$;
		\item $\kripkemodel,\world\models \neg\phi$ if $\kripkemodel,\world\not\models\phi$;
		\item $\kripkemodel,\world\models \phi_1\lor\phi_2$ if
		$\kripkemodel,\world\models\phi_1$ or $\kripkemodel,\world\models\phi_2$;
\item $\kripkemodel,\world\models\lknow\agent\phi$ if for all $\worldb$ such that $\world\relworldsa\worldb$, $\kripkemodel,\worldb\models\phi$.
	\end{itemize}
\end{definition}

Dynamic Epistemic Logic (DEL) relies on \emph{action models} (also
called ``event models'').  These models
specify how agents perceive the occurrence of an action as well as its effects on the world.

\begin{definition}
	\label{def-actionmodel}
	An \emph{action model}
	$\eventmodel = \eventmodeltuple$ is a tuple  where:
	\begin{itemize}
		\item
 $\setevents$ is a non-empty finite set of  possible \emph{actions},
		\item
 $\epistemicrelationevents{\agent}\subseteq\setevents\times \setevents$
		is the \emph{indistinguishability
			relation}  for agent $\agent$,
		\item
 $\pre: \setevents\to \languageEL$  the \emph{precondition function}, and
		\item
		 $\post: \setevents \times \AP \to
                 \languagePropositional$ is the
                 \emph{postcondition function}.
	\end{itemize}
\end{definition}

A \emph{pointed action model} is a pair $(\eventmodel,\event)$
where $\event$ represents the actual action. 
We let $|\eventmodel|$ be the \emph{size} of
$\eventmodel$, which we define as
$|\eventmodel|\egdef |\setevents|+\sum_{\agent \in \agtset} |\epistemicrelationevents\agent| + \sum_{\event\in\setevents}|\pre(\event)|+\sum_{\event\in\setevents,p\in\AP}|\post(\event,p)|$.


An action $\event$ is \emph{executable} in a world $\world$ of an
epistemic model $\epsmodel$ if
$\kripkemodel,\world\models\pre(\event)$, and in that case we define
$\postval(\world,\event):=\set{\atom \in \AP \suchthat
  \kripkemodel, \aworld \models \post(\event, \atom) }$, the set of
atomic propositions that hold after occurrence of action $\event$ in
world $\world$. 

\paragraph*{Types of actions}  An action model $\eventmodel$ is \emph{propositional} if all
pre- and postconditions of actions in $\eventmodel$ belong to
$\languagePropositional$. A \emph{public action} is a pointed action
model $\eventmodel,\event$ such that for each
agent $\agent$, $\epistemicrelationevents{\agent}$ is the identity relation.

The effect of the execution of an action in an epistemic model is
captured by the update product~\cite{baltag1998logic}:

\begin{definition}
	\label{definition:product}
	Let
        $\kripkemodel = (\setworlds,
        (\epistemicrelation\agent)_{\agent \in \agtset},
        \valuationfunction)$ be an epistemic model, and
        $\eventmodel = (\setevents,
        (\epistemicrelationevents{\agent})_{\agent \in \agtset},\pre,
        \post)$ be  an action model. The \emph{product} of
        $\kripkemodel$ and $\eventmodel$ is defined as
        $\kripkemodel \otimes \eventmodel = (\setworlds',
        (\epistemicrelation\agent)', \valuationfunction')$ where:
	\begin{itemize}
		\item $\setworlds' = \set{(\aworld, \event) \in \setworlds \times \setevents \suchthat \kripkemodel, \aworld \models \pre(\event)}$,
		\item $(\aworld,\event) \epistemicrelation\agent'
                  (\aworld',\event')$ if $\aworld
                  \epistemicrelation\agent \aworld'$ and $\event
                  \epistemicrelationevents{\agent} \event'$, and
		\item $\valuationfunction'((\aworld, \event)) = \postval(\world,\event)$.
	\end{itemize}
\end{definition} 



\newcommand{\epistemicmodelexamplepicture}{
		\begin{tikzpicture}[scale=0.6]
\tikzstyle{double_edge} = [latex'-latex',double]
\node[world, realworldarrowfromleft] (p) at (0,0) {$\world: \{p\}$};
\node[world] (notp) at (0,-2) {$\worldb: \emptyset$};
\draw[double_edge] (p) edge  node[left] {$\agenta, \agentb$} (notp);
\draw[-latex] (p) edge[loop right,  looseness=2.5] node[] {$\agenta, \agentb$} (p);
\draw[-latex] (notp) edge[loop right,  looseness=2.5] node[] {$\agenta, \agentb$} (notp);
\end{tikzpicture}}

\newcommand{\eventmodelexamplepicture}{	\begin{tikzpicture}[scale=0.8]
	\tikzstyle{double_edge} = [latex'-latex',double]
	\node[event, realworldarrowfromtop] (p) at (0,0) {$\event: \eventinfigure{p} \postconditionpfalse$};
	\node[event] (notp) at (3,0) {$\eventb: \eventinfigure{\top} \postconditiontrivial$};
	\draw[-latex] (p) edge  node[above] {$\agentb$} (notp);
	\draw[-latex] (p) edge[loop left, looseness=2] node[left] {$\agent$} (p);
	\draw[-latex] (notp) edge[loop right, looseness=2] node[right] {$\agenta, \agentb$} (notp);
	\end{tikzpicture}}

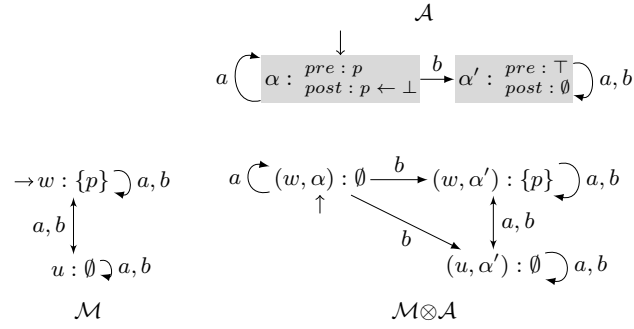
\begin{figure}
	\begin{center}
	
		\scalebox{.95}{
                  \begin{tabular}{cc}
                    &  $\eventmodel$ \\[-2mm]
			&
			\eventmodelexamplepicture \\[4mm]

			\begin{tikzpicture}[scale=0.6]
			\tikzstyle{double_edge} = [latex'-latex',double]
			\node[world, realworldarrowfromleft] (p) at (0,0) {$\world: \{p\}$};
			\node[world] (notp) at (0,-2) {$\worldb: \emptyset$};
			\draw[double_edge] (p) edge  node[left] {$\agenta, \agentb$} (notp);
			\draw[-latex] (p) edge[loop right,  looseness=2.5] node[] {$\agenta, \agentb$} (p);
			\draw[-latex] (notp) edge[loop right,  looseness=2.5] node[] {$\agenta, \agentb$} (notp);
			\end{tikzpicture}

			&

			\begin{tikzpicture}[xscale=0.8, yscale=0.6]
			\tikzstyle{double_edge} = [latex'-latex',double]
			\node[world, realworldarrowfrombottom] (p) at (0,0) {$(\world, \event): \emptyset$};
			\node[world] (p2) at (3,0) {$(\world, \eventb): \{p\}$};
			\node[world] (notp) at (3,-2) {$(\worldb, \eventb): \emptyset$};
			\draw[-latex] (p) edge  node[above] {$\agentb$} (p2);
			\draw[-latex] (p) edge  node[below] {$\agentb$} (notp);
			\draw[-latex] (p) edge[loop left, looseness=2] node[left] {$\agent$} (p);
			\draw[-latex] (notp) edge[loop right, looseness=2] node[right] {$\agenta, \agentb$} (notp);
			\draw[-latex] (p2) edge[loop right, looseness=2] node[right] {$\agenta, \agentb$} (p2);
			\draw[double_edge] (p2) edge  node[right] {$\agenta, \agentb$} (notp);
			\end{tikzpicture} \\[0mm]
                  $\kripkemodel$ &$\kripkemodel {\otimes} \eventmodel$ 
		\end{tabular}}
	\end{center}
	\caption{Example of DEL product. Symbol $\postconditiontrivial$ indicates
        the trivial postcondition that leaves valuations unchanged.}
	\label{figure:productexample}
\end{figure}

\begin{example}
  Figure \ref{figure:productexample} shows the pointed model
  $\kripkemodel, \world$ that represents a situation in which $p$ is
  true and both agents $\agent$ and~$\agentb$ do not know it. The
  pointed action model $\eventmodel, \event$ describes the action where
  agent $\agent$ learns that $p$ was true but that it is now set to
  false, while agent $\agentb$ does not learn anything (she sees action
  $\eventb$ that has trivial pre- and postcondition). In the product 
epistemic model $\kripkemodel \otimes \eventmodel, (w, \event)$,
  agent $\agent$ now knows that $p$ is false, while
  $\agentb$ still does not know the truth value of $p$, or whether agent $a$ knows it.
\end{example}

\subsection{Defining DEL games}
\label{subsec-del-games}

 We recall the definition of DEL games as introduced
 in~\cite{DBLP:conf/ijcai/MaubertPS19}, but for more general winning conditions.


The initial situation is described by an epistemic model
$\kripkemodel$, and the set of possible actions by an action model $\eventmodel$.
Because the update product in DEL can only model execution of a single
action at a time, the games that we define are turn-based. We use 
the variable $\turn$, ranging over the set of agents $\agtset$, to
represent whose turn it is to play. 
We require that 
postconditions for variable $\turn$ do not
    depend on the current epistemic situation, but instead the next value of
    $\turn$ is only determined by the
    action.
    When the precondition $\pre(\event)$ of some action
    $\event$ is satisfied, we may say that this
    action is \emph{available}.  Without loss of generality, we assume
that 
there always is at least one action
available.

The game thus starts in some world $\world\in\epsmodel$, and the agent
$\agent$ such that $\kripkemodel,\world\models \turn=\agent$ chooses
some available action $\event$ which is executed. The new epistemic
situation $\kripkemodel\otimes\eventmodel,(\world,\event)$ is given by
the update product, and the game goes on. After $n$ rounds, the
epistemic situation is described by a pointed epistemic model of the
form
$\ETLforest[n]{\epsmodel}{\eventmodel},(\world,\event_1,\ldots,\event_n)$,
where $\ETLforest[n]{\epsmodel}{\eventmodel}$ is defined by letting
$\ETLforest[0]{\epsmodel}{\eventmodel}=\epsmodel$ and
$\ETLforest[n+1]{\epsmodel}{\eventmodel}=\ETLforest[n]{\epsmodel}{\eventmodel}\otimes\eventmodel$.
In the following we may write $\world\event_1\ldots\event_n$ instead
of $(\world,\event_1,\ldots,\event_n)$, and call it a history. Given
that the model $\ETLforest[n]{\epsmodel}{\eventmodel}$ to which a
history $\world\event_1\ldots\event_n$ belongs is determined by the
length of the history, we may omit it and write, \eg,
$\world\event_1\ldots\event_n\models\phi$ instead of
$\ETLforest[n]{\epsmodel}{\eventmodel},\world\event_1\ldots\event_n\models\phi$.

In order to obtain  proper games of imperfect information,
    we will require  the following hypotheses to hold in the epistemic and
    event models defining DEL games: 

\paragraph*{Hypotheses on $\kripkemodel$ and $\eventmodel$}
  \label{hyptothesis:conditionsinputstrategyexistenceproblem}
\begin{enumerate}
\item[(H1)] 
  {\it\bf The starting player is known:}
  there is a player $\agent$ such that for all
  $\world \in \setworlds$, it holds that
  $\kripkemodel,\world \models \is\turn\agent$;
\item[(H2)] 
  {\it\bf The turn stays known:} for all
  actions $\event, \event' $ and agent $\agent$, if
  $\event \epistemicrelationequivevents{\agent} \event'$, then $\event$ and
  $\event'$ assign the same value to
  $\turn$.
\item[(H3)] 
  {\it\bf Players know their available actions:} if
  $\world\event_1\ldots\event_n\models \turn=\agent$ and
  $\world\event_1\ldots\event_n\epistemicrelationequiv{\agent}\world'\event'_1\ldots\event'_n$,
  then the same actions are available in
  $\world\event_1\ldots\event_n$ and in $\world'\event'_1\ldots\event'_n$.
\end{enumerate}

We can now define DEL games.

\begin{definition}
  \label{def-del-games}
A \emph{DEL game presentation} $(\kripkemodel,\eventmodel,\setworldsinit)$
consists of an initial epistemic model $\kripkemodel$ and an action
model $\eventmodel$ that satisfy hypotheses H1, H2 and H3, together with
a set of initial worlds $\setworldsinit$.
\end{definition}


We now describe how a DEL game presentation
$(\kripkemodel,\eventmodel,\setworldsinit)$ induces a game arena $\delgamearena$ as per Definition~\ref{def-game-arena}.

\begin{definition}
  \label{def-induced-game}
  Given a DEL game presentation $(\epsmodel,\eventmodel,\setworldsinit)$, we define
  the game arena
  $\delgamearena=(\setpos,\setposinit,\setact,\trans,\turnfunc,(\epistemicrelationequiv{\agent})_{\agent\in
    \Ag},\valuationfunction)$ where, letting
  $\ETLforest[n]{\kripkemodel}{\eventmodel}=(\setworlds^n,(\epistemicrelation{\agent}^n)_{\agent
    \in \agtset},\valuationfunction^n)$ for every $n$:
\begin{itemize}
\item $\setpos=\bigcup_{n\in\N}\setworlds^n$,
\item $\setposinit=\setworldsinit$,  
\item $\setact$ is the set of actions in $\eventmodel$,
\item $\trans(\world\event_1\ldots\event_n,\event_{n+1})=\begin{cases}
    \world\event_1\ldots\event_n\event_{n+1} &\text{if
    }\world\event_1\ldots\event_n\models\pre(\event_{n+1})\\
    \text{undefined} &\text{otherwise}
  \end{cases}
$
\item $\turnfunc(\world\event_1\ldots\event_n)=\agent$ \quad if \quad $\world\event_1\ldots\event_n\models\turn=\agent$
\item
  $\epistemicrelationequiv{\agent}\,=\,\bigcup_{n\in\N}\epistemicrelationequiv{\agent}^n$
  for each agent $\agent$, and
\item
  $\valuationfunction(\world\event_1\ldots\event_n)=\valuationfunction^n(\world\event_1\ldots\event_n)$.
\end{itemize}
\end{definition}


Observe that this game arena is infinite: the set of
  positions $\setpos$ is the set of histories. In the following
sections we will see that in some cases they admit finite
representations that we can use to decide the existence of winning
strategies.

A \emph{DEL game}
$\delgame=(\kripkemodel,\eventmodel,\setworldsinit,\win)$ consists of
a DEL game presentation $(\kripkemodel,\eventmodel,\setworldsinit)$ together with
a winning condition $\win\subseteq\setplays[{\delgamearena}]$ on the
induced game arena $\delgamearena$. 
We consider the following decision problem:

\begin{definition}[Distributed strategy synthesis for \LTLK objectives]
  \label{definition:multiepistemicgameproblemobj}
  ~
  
	\problemdefinition{A DEL game
          $\delgame=(\kripkemodel,\eventmodel,\setworldsinit,\phi)$
          with $\phi\in\LTLK$;}{Does team
            $\agentsexists$ win the game $(\delgamearena,\phi)$?}
\end{definition}

Note that the games studied in~\cite{DBLP:conf/ijcai/MaubertPS19}
correspond to the class of DEL games where the winning condition is given
by \LTLK formulas of the form $\F\phi$, where $\phi$ is purely
epistemic.

\begin{remark}
  \label{rem-uninformed}
When we evaluate whether a tuple of strategies
$(\strategy_\agent)_{\agent\in\agentsexists}$ is winning for an \LTLK
formula $\phi$, the semantics of the knowledge operators in $\phi$
does not depend on the strategies $\strategy_\agent$. 
In particular, it is not restricted to
indistinguishable histories that follow these strategies,
but instead it considers all indistinguishable histories in the game.
This semantics models situations in which agents do not know the strategies of
agents $\agentsexists$, and it is the one also used
in~\cite{DBLP:conf/ijcai/MaubertPS19} 
 and in DEL epistemic
planning~\cite{DBLP:journals/jancl/BolanderA11,DBLP:conf/ijcai/AucherB13,bolander2015complexity,DBLP:journals/corr/Bolander17},
where the agents do not know which plan is being executed.
This semantics is
called~\emph{uninformed semantics} in~\cite{DBLP:conf/kr/MaubertM18},
contrary to the
\emph{informed} one.
See also~\cite{DBLP:conf/mfcs/Puchala10,DBLP:phd/hal/Maubert14} for more discussions on the matter.
\end{remark}

\subsection{Discussion on initial positions}
\label{sec-particular-cases}

One subtlety that arises when formalising existence of winning strategies under imperfect
information is in defining what having a winning strategy means.
For instance, are we satisfied with the agents in $\agentsexists$ having a distributed strategy 
that is winning from the initial position of the game, even if they do
not \emph{know} that it is winning, in the sense that there is a world
that some agent in $\agentsexists$ considers like a possible initial
position and from which the distributed strategy is not winning?
Or instead do we want everybody in the team to know that the team's
strategy is winning? 
These two notions have sometimes been  called \emph{objective winning}
and \emph{subjective winning}, respectively
(see~\cite{bulling2014comparing} and
also~\cite{DBLP:journals/jancl/JamrogaA07} for similar considerations). 
We could also ask whether there is \emph{distributed knowledge} or
\emph{common knowledge}~\cite{Fagin95knowledge} that the distributed strategy is winning.

Note that we can model all of these notions by tuning the set of initial
worlds in the definition of a DEL game, as this defines the set of
outcomes that we consider, \ie, the set of
plays from which the strategies should be winning. Assume we have an
initial epistemic model $\kripkemodel$ with an
 initial world $\world_\init$. If we are interested in distributed
strategies that are objectively winning from $\world_\init$, we simply
set $\setworldsinit=\{\world_\init\}$ in the DEL game. If instead 
we want subjectively winning strategies, \ie,  strategies that
not only are winning, but such that everybody in the team
$\agentsexists$ knows that they are winning, then we let the set of
initial worlds in the DEL game be
\[\setworldsinitsubj=\{\world\in\kripkemodel
\mid \world\epistemicrelationequiv{\agent}\world_\init \mbox{ for some
}\agent\in\agentsexists\}.\]



Objective distributed strategy synthesis was studied
in~\cite{DBLP:conf/ijcai/MaubertPS19} for reachability epistemic
objectives, \ie, when an epistemic objective is given as an epistemic
formula $\phi\in\languageEL$, and $\win$ is defined as the sets of
plays $\play\in\Play{\forest}$ for which there exists $i\in\N$ such
that $\forest,\pi_{\leq i}\models \phi$. Note that such winning
conditions can be specified by \LTLK formulas $\F\phi$. It is shown that objective 
distributed strategy synthesis is undecidable for propositional
actions, already for a team of two players and reachability epistemic
objectives.
Since the problem we study is more general, we inherit this
undecidability result.

\begin{theorem}
  \label{theo-undecidable}
  Distributed strategy synthesis for \LTLK objectives is undecidable
  already for propositional actions and formulas of the form $\F\phi$,
  where $\phi$ is purely epistemic.
\end{theorem}

In the rest of the paper we describe various cases in which decidability can be recovered.




\section{DEL games with propositional actions}
\label{sec-propositional}

Extending a result
from~\spchanged{\cite{DBLP:phd/hal/Maubert14,DBLP:journals/corr/AucherMP14}} in
the case of DEL planning, it was proved
in~\cite{DBLP:conf/ijcai/MaubertPS19} that  in
the case of propositional actions, games generated by DEL
game presentations are regular, and one can compute an equivalent
finite game arena:

\begin{proposition}[\cite{DBLP:conf/ijcai/MaubertPS19}]
  \label{prop-regular-imp}
  Given a DEL game presentation
  $(\kripkemodel,\eventmodel)$ where
  $\eventmodel$ is propositional, one can
 construct a finite game arena $\gamestruct$ equivalent to
 $\delgamearena$ such that
 $|\gamestruct|\le |\epsmodel|+|\eventmodel|\times 2^{m}$, where $m$ is
the number of atomic propositions in $\kripkemodel$ and $\eventmodel$.
\end{proposition}

From the latter result, if the winning condition $\win$ is
given as a formula $\phi\in\LTLK$, then the same winning condition on
$\gamestruct$ yields a multiplayer epistemic game that is equivalent to the original DEL
game $\delgame$ in terms of existence of distributed winning strategies.
And because such games can be decided in the case of hierarchical
information, we obtain our result.

More precisely, we say that a DEL game $(\epsmodel,\setworlds_\init,\eventmodel,\win)$ presents 
\emph{hierarchical information} if  the set of agents
$\agentsexists$ can be totally ordered ($\agent_1 < \ldots <\agent_n$) so that  $\epistemicrelationequiv{\agent_i}\subseteq\;
\epistemicrelationequiv{\agent_{i+1}}$ and $\epistemicrelationequivevents{\agent_i}\subseteq\;
\epistemicrelationequivevents{\agent_{i+1}}$, for each $1\le
i<n$.

\begin{theorem}[\cite{DBLP:conf/mfcs/Puchala10,DBLP:conf/kr/MaubertM18}]
  \label{theo-unif-strat-KR}
In  multiplayer epistemic
  games 
  with hierarchical information and epistemic temporal objectives, 
distributed strategy synthesis  
  is
  decidable \bmchanged{for the uninformed semantics}.
\end{theorem}

Proposition~\ref{prop-regular-imp} together with
Theorem~\ref{theo-unif-strat-KR} imply that:

\begin{theorem}
  \label{theo-propositional}
Distributed strategy synthesis for \LTLK objectives with propositional
  actions and hierarchical information is decidable.
\end{theorem}

\begin{remark}
  \label{rem-multiple-init}
  The results
  in~\cite{DBLP:conf/mfcs/Puchala10,DBLP:conf/kr/MaubertM18} are
  established for games with a unique initial position, \ie when
  $\setpos_\init$ is a singleton $\{\posinit\}$. However it is
  easy to see that distributed synthesis with
  multiple initial positions $\setpos_\init$ can be reduced to the
  case of a unique initial position: one only needs to add a fresh
   position $\pos_\init$ that is used as initial position, from which
   one can attain all positions in $\setpos_\init$, and only these. It does not matter who
   $\pos_\init$ belongs to or how the agents observe it.
\end{remark}


\section{DEL games with public actions}
\label{sec:publicactions}
In this section, we show that when all actions are public, the
distributed strategy synthesis problem 
is decidable for \LTLK winning conditions.  Towards this end, we first prove a result similar to
Proposition~\ref{prop-regular-imp}: we show that given a DEL game
presentation $(\kripkemodel,\eventmodel,\setworldsinit)$ where all actions
$\event\in\eventmodel$ are public, the infinite game arena
$\delgamearena$ is regular and can be folded back into a finite game
arena. This allows us to reduce the distributed strategy synthesis
problem to a distributed synthesis problem on explicit game arenas, for which a solution is known in the case of public actions
and \LTLK objectives~\cite{DBLP:conf/concur/MeydenW05}.

Note that the decidability result for reachability DEL games with
public actions  in~\cite{DBLP:conf/ijcai/MaubertPS19} does not rely on
this kind of construction, but instead is proved
by  providing a direct alternating algorithm. There is a problem in the
way this algorithm forces strategies to be uniform.
In the case of public actions and unique initial world considered
there it can be easily corrected, as there
is in fact no need to check for uniformity of strategies (see
Remark~\ref{rem-unique-init}). But for our more general setting with
multiple initial worlds, this approach was not sound.


\begin{proposition}
  \label{theo-regular-public}
    Given a DEL game presentation
  $(\kripkemodel,\eventmodel,\setworldsinit)$ where
all actions in  $\eventmodel$ are public, one can
compute a finite game arena $\gamestruct$  
equivalent to
 $\delgamearena$ such that
 $|\gamestruct|\le m(2^p+1)^{m}$, with $p$ the
 number of atomic propositions in $(\kripkemodel,\eventmodel)$ and $m$  the number of
 worlds in $\kripkemodel$.  
\end{proposition}

\begin{proof}
For every position $\world\event_1\ldots\event_n$
in $\delgamearena$ we define its \emph{attached epistemic model} $\kripkemodel_{\world\event_1\ldots\event_n}$ as the
connected component of $\ETLforest[n]{\kripkemodel}{\eventmodel}$ that
contains $\world\event_1\ldots\event_n$.
Since all actions in $\eventmodel$ are public,
 for all positions $\world\event_1\ldots\event_n$ and
$\world\event_1\ldots\event_n\event_{n+1}$ in $\delgamearena$
we have that $\kripkemodel_{\world\event_1\ldots\event_{n+1}}$ is no
bigger than $\kripkemodel_{\world\event_1\ldots\event_n}$:
indeed,
the application of a public action can only remove worlds from
$\kripkemodel_{\world\event_1\ldots\event_n}$ (those that do not satisfy the precondition) and change
the valuations of the remaining worlds.
As a result
there is only a finite number
 of different positions
$\world\event_1\ldots\event_n$
in $\delgamearena$, up to isomorphism of their attached models. 
We write  $\equivalence$ the  equivalence relation on positions of
$\delgamearena$ defined by letting two positions
be equivalent if their attached models are isomorphic, and we let
$\eqclass{\world\event_1\ldots\event_n}$ be the equivalence class of
position $\world\event_1\ldots\event_n$ for this relation.

Let us write $\delgamearena =(\setpos,\setposinit,\setact,\trans,
\turnfunc,\{\relequiva\}_{\agent\in\agtset}, \valuationfunction)$.
The finite game arena $\gamestruct$ is the quotient of  $\delgamearena$ with $\equivalence$. More precisely, 
$\quotient{\delgamearena} =(\setpos',\setposinit',\setact',\trans',
\turnfunc',\{\relequiva'\}_{\agent\in\agtset}, \valuationfunction')$, where:
\begin{itemize}
\item $\setpos'=\{\eqclass{\world\event_1\ldots\event_n}
  \mid \world\event_1\ldots\event_n\in\setpos\}$,
\item $\setposinit'=\{\eqclass{\world}\mid \world\in\setposinit\}$,
\item $\setact'=\setact$,
\item $\trans'(\eqclass{\pos},\event)=
  \begin{cases}
    \eqclass{\trans(\pos,\event)} &\mbox{if }\trans(\pos,\event)\mbox{
      is defined}\\
    \mbox{undefined} &\mbox{otherwise,}
  \end{cases}
$ 
\item $\eqclass{\pos} \relequiva' \eqclass{\pos'}$ if $\pos\relequiva
  \pos'$, and
\item $\valuationfunction'(\eqclass{\pos})=\valuationfunction(\pos)$.
\end{itemize}

To see that $\trans'$ is well defined, observe that if
$\pos\equivalence\pos'$ then $\trans(\pos,\event)$ is defined if,
and only if, so is $\trans(\pos',\event)$, and in this case $\trans(\pos,\event)\equivalence\trans(\pos',\event)$.
The fact that $\relequiva'$ and $\valuationfunction'$ are well defined
follows directly from isomorphism of attached models.

To construct $\quotient{\delgamearena}$, one can enumerate all
possible attached models $\kripkemodelb$ (modulo isomorphism) as
follows: for each world $\world$ in the original model $\kripkemodel$,
decide first whether there is some world of
the form $\world\event_1\dots\event_n$  in
$\kripkemodelb$; if there is, there is only one, because all actions
are public, and thus any position of the form
$\world\event'_1\ldots\event'_n$ (with for some $i$, $\event'_i \neq \event_i$) is not related to
$\world\event_1\dots\event_n$ and thus does not appear in $\kripkemodelb$. Then, one chooses the valuation over the
atomic propositions involved in the problem for each world in $\kripkemodelb$. Indistinguishability
relations are inherited from $\kripkemodel$:
$\world\event_1\dots\event_n\relequiva \world'\event_1\dots\event_n$
if, and only if, $\world\relequiva\world'$. The number of such
different attached models is  bounded by
$\sum_{k=1}^{m}\binom{m}{k}2^{pk}=(2^p+1)^m$, and each one
  has at most $m$ worlds. We thus have at most $m(2^p+1)^m$ positions
in $\setpos'$. \bmchanged{It remains to build the function $\trans'$ 
as described in the definition of $\quotient{\delgamearena}$: to
determine $\trans'(\eqclass{\pos},\event)$, compute the product of the
representant of $\eqclass{\pos}$ in $\quotient{\delgamearena}$ with
$(\eventmodel,\event)$, and identify the only position of
$\quotient{\delgamearena}$ that is isomorphic to the result.
Testing for isomorphism can be done in exponential time, and there is
an exponential number of positions to test. Finally, the whole construction can be done in exponential time.}
\end{proof}


Proposition~\ref{theo-regular-public} ensures that from a DEL game presentation
$(\kripkemodel,\eventmodel)$ with public actions we can construct an
equivalent finite game arena of exponential size. Moreover, in this game arena, all actions are public in the sense
of~\cite{DBLP:conf/atal/BelardinelliLMR17}. In this latter work, 
model checking \ATLs with epistemic operators (\ATLsK) on game arenas with
public actions is proved in \TWOEXPTIME. More precisely, the proposed procedure
takes time doubly exponential in the size of the formula, but only
exponential time in the size of the game structure. Combined with our exponential
construction from Theorem~\ref{theo-regular-public}, we obtain a
procedure to solve our distributed strategy synthesis problem for public
actions in doubly exponential time, both in the size of the DEL game
presentation and in the size the \LTLK winning condition $\phi$.

To make our argument more precise, we briefly recall the syntax and semantics of \ATLsK.
The syntax of \ATLsK is given by the following grammar:
\begin{align*}
  \phi & ::= p \mid \lnot \phi \mid \phi \vee \phi \mid \K \phi \mid \Estrat\psi \\
  \psi & ::= \phi \mid \lnot \psi \mid \psi \vee \psi \mid \next\psi \mid \psi \until \psi 
\end{align*}
  
where $p\in\AP$ and $\agent\in\agtset$.
Formulas of type $\phi$ are called \emph{history formulas}, while
those of type $\psi$ are called \emph{path formulas}. 
Note that, in addition, the authors of \cite{DBLP:conf/atal/BelardinelliLMR17} consider epistemic operators for common and distributed knowledge. We omit them from the syntax as we do not consider such operators in this work.

The semantics of \ATLsK is defined
in~\cite{DBLP:conf/atal/BelardinelliLMR17} on concurrent-game
structures. We instead define it on our turn-based game structures,
which can be seen as a particular case.
Let us first recall the notion of public actions for game arenas
considered in~\cite{DBLP:conf/atal/BelardinelliLMR17}:
\begin{definition}
  \label{def-public-actions-arenas}
A game
arena $\gamestruct=(\setpos,\setposinit,\setact,\trans,\turnfunc,(\epistemicrelationequiv{\agent})_{\agent\in
  \Ag},\valuationfunction)$ has \emph{only public actions} if, for
all $\pos,\pos'\in\setpos$ and $\act,\act'\in\setact$ such that
$\act\neq\act'$, we have
$\trans(\pos,\act)\not\epistemicrelationequiv{\agent}\trans(\pos',\act')$. 
\end{definition}

The semantics of \ATLsK formulas is defined with respect to a game
arena $\gamestruct$ together with a history $\hist$ in case of a
history formula, or a play $\play$ and a point in time $i\in\N$ in
case of a path formula.
\begin{align*}
  \gamestruct,\hist &\models p &&\ifdef p\in \valworlds(\hist)\\
  \gamestruct,\hist &\models \lnot \phi &&\ifdef\gamestruct,\play,i \not \models \phi\\
  \gamestruct,\hist &\models \phi_1 \vee \phi_2
                                 &&\ifdef\gamestruct,\play,i \models
                                    \phi_1 \; \mbox{ or }\;\gamestruct,\play,i \models \phi_2\\
  \gamestruct,\hist &\models \K\phi &&\ifdef \forall
                                            \hist'\in\sethistories \text{ s.t. }
                                            \hist' \relequiva
                                         \hist,\;\gamestruct,\hist'
                                       \models \phi \\
    \gamestruct,\hist &\models  \Estrat \psi &&\ifdef \text{there
                                                exists }\disstrategy[\coal]
                                                \text{
                                                s.t. }\forall\play\in\outcome(\hist,\disstrategy[\coal]), \\
& &&  \;\;\;\;\;\; \gamestruct,\play,|\hist|-1 \models \psi\\
  \gamestruct,\play,i &\models \phi &&\ifdef \gamestruct,\play_{\leq
                                       i}\models \phi\\
  \gamestruct,\play,i &\models \lnot \phi &&\ifdef\gamestruct,\play,i \not \models \phi\\
  \gamestruct,\play,i &\models \phi_1 \vee \phi_2
                                 &&\ifdef\gamestruct,\play,i \models
                                    \phi_1 \; \mbox{ or }\;\gamestruct,\play,i \models \phi_2\\
  \gamestruct,\play,i &\models \next\phi &&\ifdef\gamestruct,\play,i+1 \models \phi\\
  \gamestruct,\play,i &\models \phi_1 \until \phi_2
                                 &&\ifdef\exists i'\geq i
                                    \;\text{ s.t.
                                    }\;\gamestruct,\play,i'
                                           \models \phi_2 \;\text{
                                    and }\\
                                          & && \quad\;\;\forall i'' \;\text{ s.t. }\; i\leq i''<i', \gamestruct,\play,i'' \models \phi_1
\end{align*}

In the above definition, $\outcome(\hist,\disstrategy[\coal])$ is the set of
plays that extend $\hist$ by following $\disstrategy[\coal]$, which
corresponds to the objective semantics discussed in
Section~\ref{sec-particular-cases}.  Thus, it is easy to see that for
an \LTLK formula $\psi$ and a game arena $\gamestruct$ with a
singleton set of initial positions $\setposinit=\{\posinit\}$, it
holds that $\agentsexists$ wins $(\gamestruct,\phi)$ if, and only if,
$\gamestruct,\posinit\models\Estrat[\agentsexists]\psi$  (technically,
$\psi$ should be modified by replacing each occurrence of knowledge
modality $\K$ by $\K\Estrat[\emptyset]$ in order to have a history
formula).

As explained in~\cite{DBLP:conf/atal/BelardinelliLMR17}, their procedure
can be adapted to the subjective semantics, \ie when the distributed
strategy $\disstrategy[\coal]$ is required to be winning from all positions
that are equivalent to $\posinit$ for some agent in
$\agentsexists$. It is not hard to see that this adaptation
would also work for any set $\setposinit$ of initial positions, and
that it does not change the complexity of the procedure.

 \begin{theorem}
\label{theo-publicactions}
For public actions, distributed strategy synthesis for \LTLK objectives is \TWOEXPTIME-complete.
 \end{theorem}

 \begin{proof}
   Let $(\kripkemodel,\eventmodel,\setworldsinit,\phi)$ be a DEL game
   with $\phi\in\LTLK$ and such that all actions 
 are public. By Proposition~\ref{theo-regular-public} we
   can compute in exponential time a finite game arena $\gamestruct=\quotient{\delgamearena}$
   equivalent to $\delgamearena$ with   $|\gamestruct|\le
   m(2^p+1)^{m}$, where $m$ is the number of worlds in $\kripkemodel$
   and $p$ is the number of atomic propositions in
   $(\kripkemodel,\eventmodel)$.

We show that this finite game arena $\gamestruct$ has only
   public actions in the sense of
   Definition~\ref{def-public-actions-arenas}.
   Write $\gamestruct =(\setpos,\setact,\trans,
\turnfunc,\{\relequiva\}_{\agent\in\agtset}, \valuationfunction)$, and
recall that $\setpos=\{\eqclass{\world\event_1\ldots\event_n}
  \mid \world\event_1\ldots\event_n\in\delgamearena\}$. 
  Take two positions $\pos=\eqclass{\world\event_1\ldots\event_i}$ and
  $\pos'=\eqclass{\world'\event'_1\ldots\event'_j}$ in $\setpos$, and
  $\event,\event'\in\setact$ such that $\event\neq\event'$. By
  definition,
  $\trans(\pos,\event)=\eqclass{\world\event_1\ldots\event_i\event}$
  and
  $\trans(\pos',\event')=\eqclass{\world'\event'_1\ldots\event'_j\event'}$. Since
$\event$ and $\event'$ are public actions and  $\event\neq\event'$, necessarily 
  $\world\event_1\ldots\event_i\event\not\epistemicrelation{\agent}\world'\event'_1\ldots\event'_j\event'$,
  entailing 
  $\trans(\pos,\event)\not\epistemicrelation{\agent}\trans(\pos',\event')$.

We can thus use the model-checking procedure
from~\cite{DBLP:conf/atal/BelardinelliLMR17} to evaluate whether
$\gamestruct$ satisfies
$\Estrat[\agentsexists]\phi$. This procedure takes time doubly
exponential in the size of $\phi$ and exponential in
the size of $\gamestruct$, which is itself exponential in the size of
the DEL game presentation $(\kripkemodel,\eventmodel)$, hence the
upper bound. The lower bound is obtained by reduction from
LTL synthesis, which is
\TWOEXPTIME-complete~\cite{pnueli1989synthesisshort}.
 \end{proof}

%

 \begin{remark}
   \label{rem-unique-init}
We point out that
in the context of public actions, the case of a unique initial
position is in
fact equivalent to the perfect-information simplification of the
problem, in which the uniformity requirement for strategies is
dropped. Indeed in this case the uniformity constraint
is trivial to satisfy: assume there is a winning distributed strategy
$(\strategy_\agent)_{\agent\in\agtset}$ where the strategies are not
necessarily uniform.
Take two histories $\hist$ and $\hist'$ that are
equivalent to some agent $\agent\in\agentsexists$. If one of them, say $\hist'$,
does not start in $\world_\init$, then it does not matter how strategies 
$(\strategy_\agent)_{\agent\in\agtset}$ are defined on 
$\hist'$, because they are not required to be winning from worlds
other than $\world_\init$; one can thus change the definition of the
strategies on $\hist'$ to make them uniform. Otherwise, if both start in
$\world_\init$, because actions are public, 
$\hist=\hist'$ so that the strategies are already uniform on these
histories.
\end{remark}



\section{DEL games with public announcements}
\label{sec:publicannouncements}

We now investigate DEL games with  \emph{public
  announcements}, which are public actions with
 no effect besides epistemic ones. We
assume that the winning conditions are restricted to
$\fragmentLTLKUnonesting$, the syntactic fragment of \LTLK objectives
with no next modality ($\next$) and with no temporal operator
($\next$ or $\until$) under the scope of a knowledge modality
($\K$). We also assume that the games are round-robin, \ie the turn
goes from an agent to the next in a circular order, and we assume a
unique initial world.
We show that, in this context, deciding the existence of a
winning strategy for the team $\agentsexists$ in a DEL game
$\delgame=(\kripkemodel,\eventmodel,\setworldsinit,\phi)$ and $\phi \in \fragmentLTLKUnonesting$
is \PSPACE-complete. 
Because reachability goals are definable in
$\fragmentLTLKUnonesting$, this result generalises the
\PSPACE-completeness result established in~\cite{DBLP:conf/ijcai/MaubertPS19}.
 
Formally, a \emph{public announcement} is a public action
$(\eventmodel,\event)$ such that $\post(\event, p) = p$, for each
variable $p$ but variable $\turn$. This is the natural generalisation
of public announcements as defined in
\cite{DBLP:journals/synthese/Plaza07,DitmarschvdHoekKooi}.
 As a consequence on the product update, 
either an announcement is \emph{non-informative} and the updated epistemic model remains the same (modulo variable $\turn$), or it is
\emph{informative} and yields an epistemic model with strictly
less worlds.

\begin{theorem}
  \label{theorem:objpubannoun}
In round-robin DEL games with unique initial world and  public
announcements, distributed strategy synthesis for   \fragmentLTLKUnonesting winning conditions is \PSPACE-complete.
\end{theorem}

The rest of this section is dedicated to the proof of Theorem~\ref{theorem:objpubannoun}.
\newcommand{\strategyeager}{{\strategy}^{\text{eager}}}
	\newcommand{\histrep}[1]{{look}_\agent({#1})}
        The problem is already \PSPACE-hard for reachability goals~\cite{DBLP:conf/ijcai/MaubertPS19}, therefore it is
        still \PSPACE-hard for \fragmentLTLKUnonesting objectives.
%
%
        Regarding the membership in \PSPACE, the two main ideas are:
        \begin{enumerate}
        \item From an initial epistemic model
          $\kripkemodel = (\setworlds,
          (\epistemicrelation{\agent})_{\agent \in \agtset},
          \valuationfunction)$, there are at most $|\setworlds|$
          informative announcements;
        \item To limit the length of plays, we can shorten, as depicted in Figure~\ref{figure:publicannouncementbypass}, 
          sequences of non-informative announcements: from a
          strategy $\strategy$
          we show how to extract an 
          \emph{eager}  strategy $\strategyeager$ that  performs all informative announcements
          eventually recommended by $\strategy$ as early as
          possible.
          Thus, any sequence of non-informative announcements followed
          by an informative one is of length at most the number $|\agtset|$
          of agents: if an agent wants to perform an informative event
          in the future, she can do so as soon as it is her turn to
          play. This, in a round-robin game, happens in at most
          $|\agtset|$ steps.
        \end{enumerate}
As a result of these two points, we can search for eager strategies via 
          a depth-first search in $\delgamearena$ up to depth
          $|\agtset| \times |\modelM|$.

          We now describe how to extract eager strategies.  In the
          following we call \emph{states} the attached epistemic
          models\footnote{see proof of
            Proposition~\ref{theo-regular-public},
            page~\pageref{theo-regular-public}.} in $\delgamearena$,
          writing them $s, s_1, \dots$. These are mere submodels of
          the initial model $\modelM$, if we ignore variable
          $\turn$. We then write $s^k$ for the sequence with $k$
          consecutive $s$'s (only variable $\turn$ is
          changing). 

          Given a distributed strategy
          $\strategy = (\strategy_\agent)_{\agent\in\agentsexists}$,
          we let the eager distributed strategy
          $\strategyeager=(\strategyeager_\agent)_{\agent\in\agentsexists}$
          be defined by
          $\strategyeager_\agent(\hist) :=
          \strategy_\agent(\histrep{\hist})$ where $\histrep{\hist}$
          is a history called \emph{look ahead}. This
          $\histrep{\hist}$ is, when it exists, a history that follows
          $\strategy$, in which it is $\agent$'s turn to play and
          $\strategy_\agent(\histrep{\hist})$ is informative. 
          Also, $\hist$ is a stuttering-equivalent subsequence
          of $\histrep{\hist}$ where agents  bypass non-informative announcements and perform
          informative ones prescribed by $\strategy$ as soon as possible. There might
          be no such $\histrep{\hist}$ if agents are not eager in
          $\hist$. 
%
      We define $\histrep{\hist}$ 
      by induction:
	\begin{itemize}
		\item $\histrep{\epsilon} := \epsilon$ (base case);
		\item if $\hist=\hist' s$ where $\hist'$ is a
 history and $s$ is a state  with either
 $\hist'=\epsilon$ or $|s|<|\last(\hist')|$ (an informative
 announcement has been made), $\histrep{(\hist' s^k)}=\histrep{\hist'} s^\ell$, such that $\histrep{\hist'} s^\ell$ follows
		$\strategy$ and $\ell$ is
                  \begin{enumerate}
                \item[i.]  if it exists, the smallest integer such
                that $\turn(\histrep{\hist'} s^\ell)=\agent$, and
$\strategy(\histrep{\hist'} s^\ell)$ is informative at $\histrep{\hist'} s^\ell$.
\item[ii.]             otherwise take $\ell = k$.
                  \end{enumerate}
      
	\end{itemize}
	
	\begin{lemma}
		\label{lemma:publicannouncements-form}
Any outcome of $\strategyeager$ is of the form $s_1^{k_1} \dots s_n^{k_n} s^\omega$ where $|s_1|<|s_2|< \dots <|s_n|$ and $k_i < |\agtset|$.
	\end{lemma}

As informative announcements prescribed by
        $(\strategyeager_\agent)_{\agent\in\agentsexists}$ coincide with the ones prescribed by  $(\strategy_\agent)_{\agent\in\agentsexists}$,
        the outcomes of $\strategyeager$ are \emph{stuttering
          equivalent} to some outcome of $\strategy$.  Recall that two
        paths are stuttering equivalent if  omitting repetitions of
        states in both of them yields the same sequence of states. For instance,
 $s_1 s_2 s_2 s_3$ and $s_1 s_1 s_2 s_3 s_3$ are stuttering equivalent
 (see~\cite{lamport1983good}).
	
	\begin{lemma}
		\label{lemma:publicannouncements-size}
		For any outcome of $\strategyeager$, there exists a stuttering equivalent outcome of $\strategy$.
	\end{lemma}

        We now design a polynomial space algorithm that
        decides whether there exists such an eager strategy
        $\strategyeager_\agent$ by performing a depth-first-search (minmax-like approach)
      in the unfolding of $\delgamearena$ at polynomial depth $|\agtset| \times |\modelM|$.

 	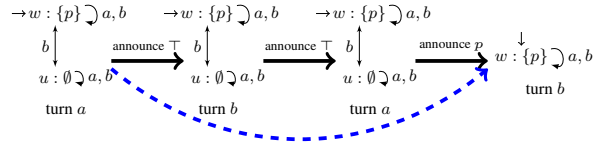
\begin{figure}
 		\newcommand{\epistemicmodelexamplepicturesmallturn}[1]{\scalebox{0.7}{
 				\begin{tabular}{c}		
 					\begin{tikzpicture}[scale=0.6]
 					\tikzstyle{double_edge} = [latex'-latex',double]
 					\node[world, realworldarrowfromleft] (p) at (0,0) {$\world: \{p\}$};
 					\node[world] (notp) at (0,-2) {$\worldb: \emptyset$};
 					\draw[double_edge] (p) edge  node[left] {$\agentb$} (notp);
 					\draw[-latex] (p) edge[loop right,  looseness=2.5] node[] {$\agenta, \agentb$} (p);
 					\draw[-latex] (notp) edge[loop right,  looseness=2.5] node[] {$\agenta, \agentb$} (notp);
 					\end{tikzpicture} \\
 					turn \ensuremath{#1}
 		\end{tabular}}}
		
 		\newcommand{\smallepistemicmodelexamplepicturesmallturn}[1]{\scalebox{0.7}{
 				\begin{tabular}{c}		
 					\begin{tikzpicture}[scale=0.6]
 					\tikzstyle{double_edge} = [latex'-latex',double]
 					\node[world, realworldarrowfromtop] (p) at (0,0) {$\world: \{p\}$};
 					\draw[-latex] (p) edge[loop right,  looseness=2.5] node[] {$\agenta, \agentb$} (p);
 					\end{tikzpicture} \\
 					turn \ensuremath{#1}
 		\end{tabular}}}
	
 	\newcommand{\backgroundwithstrategy}{\draw[draw=none,fill=yellow!20!white] (-1, 1) rectangle (7, -1);}

 		\tikzstyle{announcement} = [->, line width=0.5mm]
 		\newcommand{\announce}[1]{\text{\tiny announce \ensuremath{#1}}}
 		\begin{center}
 			\begin{tikzpicture}
 			\node[] (0) at (0, 0) {\epistemicmodelexamplepicturesmallturn a};
 			\node[] (1) at (2, 0) {\epistemicmodelexamplepicturesmallturn b};
 			\node[] (2) at (4, 0) {\epistemicmodelexamplepicturesmallturn a};
 			\node[] (3) at (6.3, 0) {\smallepistemicmodelexamplepicturesmallturn b};
 			\foreach \x in {0, 1, 2, 3} {
 				\node (\x s) at (2*\x+0.5, 0) {};
 				\node (\x t) at (2*\x-0.3, 0) {};
 			}
 			\draw[announcement] (0s) edge node[above] {$\announce \top$} (1t);
 			\draw[announcement] (1s) edge node[above] {$\announce \top$} (2t);
 			\draw[announcement] (2s) edge node[above] {$\announce p$} (3t);
 			\draw[announcement, dashed, blue] (0s) edge[bend right=40] node[below] {$\announce p$} (3t);
 			\end{tikzpicture}
 			\vspace{-1cm}
%
%
 		\end{center}
 		\caption{Strategy $\strategyeager$ bypasses the
                   non-informative announcements in $\strategy$ and
                   makes the same informative announcement (here $p$) but
                   eagerly.\label{figure:publicannouncementbypass}}
 	\end{figure}

      Every time a leaf $s$ is reached, it is considered as the
      attached epistemic model in which the game stays forever with no
      more informative announcement (\ie $s^\omega$). We then evaluate
      the winning condition $\fragmentLTLKUnonesting$-formula by model
      checking the path carried by the branch in this tree. Now, model
      checking a path against $\fragmentLTLKUnonesting$ is a problem
      in \PTIME: because we require that no temporal operator occur
      under the scope of knowledge modalities, epistemic subformulas
      occurring in the challenged $\fragmentLTLKUnonesting$-formula
      can be evaluated locally on the path so that these subformulas become
      mere propositions. It remains to model check an \LTL-formula on this
      marked path which can be done in polynomial time 
      (see for example \cite[Section 6.4.3]{demri2016temporal}).
	
      Notice that while running this depth-first-search, one needs to
      remember the current branch (needed for backtracking in the
      minmax algorithm) as well as the information used by the
      $\fragmentLTLKUnonesting$ path model-checking procedure, which
      yields a poly-size information, so that the algorithm runs in
      polynomial space.

      We now prove that this algorithm is correct. If the algorithm
      accepts the input, then we have some winning strategy in hands,
      namely some $\strategyeager$, and we are done.
      
      Conversely, assume there exists a winning
      strategy $\strategy$.
      Because any \LTLU-formula  (no $\next$ operator) is stuttering-invariant (see
      for example \cite[Th.\ 6.6.5 p.\ 184]{demri2016temporal}), and
      because in our logic \fragmentLTLKUnonesting, epistemic
      subformulas are evaluated locally in states, just as
      propositions, the outcomes of $\strategyeager$ do also satisfy
      the winning conditions by Lemma
      \ref{lemma:publicannouncements-size}. Now because of
      Lemma~\ref{lemma:publicannouncements-form}  strategy
      $\strategyeager$ will be found by the
      algorithm, which concludes.

\section{Conclusion}
\label{sec:conclusion}




We generalised the setting defined
in~\cite{DBLP:conf/ijcai/MaubertPS19} for distributed synthesis in DEL
games, moving from reachability winning conditions to ones expressed
in \LTLK, and  allowing for multiple initial positions, which
allows us to capture various semantics of strategic ability but also
makes the problem harder in the case of public actions.

We showed that the main results established
in~\cite{DBLP:conf/ijcai/MaubertPS19} can be lifted to this more
general setting: of course the problem remains undecidable, but
decidability is retrieved in the case of public actions, as well as propositional actions 
together with hierarchical information.

In the latter case the problem is, as usual, nonelementary, as
each agent in the team with a different observation of the game adds an
exponential to the cost of solving it~\cite{PR90,peterson2002decision}. But for public actions
we proved that the problem is in \TWOEXPTIME, which is optimal as this
is already the complexity of solving \LTL synthesis~\cite{pnueli1989synthesisshort}.
A central technical result was to establish the regularity
of infinite DEL game arenas generated from public actions.
We conjecture that our techniques could extend to even more
expressive winning conditions, such as ones expressible in epistemic mu-calculus. 

Regarding public announcements, we showed that 
the distributed synthesis problem is PSPACE-complete for winning
conditions in the fragment \fragmentLTLKUnonesting, when games are
round-robin and have a unique initial world.
The complexity
of generalisations such as several initial positions or  winning
conditions beyond \fragmentLTLKUnonesting is still open.




\end{document}